\documentclass[letterpaper, 10 pt, conference]{ieeeconf}  

\IEEEoverridecommandlockouts                              
\usepackage{float}
\usepackage{color}
\usepackage{multirow}
\usepackage{balance}
\usepackage{placeins}
\usepackage{varioref}
\let\labelindent\relax
\usepackage{enumitem}
\usepackage{url}
\usepackage{amsmath,amssymb,amsfonts}
\usepackage{algorithmic}
\usepackage{algorithm}
\usepackage{textcomp}
\usepackage[pdftex]{graphicx}
\usepackage{subcaption}
\graphicspath{{figs/}}
\usepackage{epstopdf}
\usepackage{booktabs}
\usepackage{amsmath}
\usepackage{mathtools}

\usepackage{amsthm}
\newtheorem{theorem}{Theorem}
\newtheorem{lemma}{Lemma}
\newtheorem{case}{Case}

\newtheorem{remark}{Remark}

\theoremstyle{definition}
\newtheorem{definition}{Definition}[]

\title{\LARGE \bf
On the Spatial Pattern of Input-Output Metrics for a Network Synchronization Process
}

\author{Subir Sarker and Sandip Roy
\thanks{The authors acknowledge support from ARPA-E, Pacific Northwest National Laboratory, and National Science Foundation.}
\thanks{The authors are with the School of Electrical Engineering and Computer Science, Washington State University, Pullman, WA 99163, USA.
        {\tt\small \{subir.sarker,sandip\}@wsu.edu}}%
}

\begin{document}

\maketitle
\thispagestyle{empty}
\pagestyle{empty}

\begin{abstract}

A graph-theoretic analysis is undertaken for a compendium of input-output (transfer) metrics of a standard discrete-time linear synchronization model, including $l_p$ gains, frequency responses, frequency-band energy, and Markov parameters. We show that these transfer metrics exhibit a spatial degradation, such that they are monotonically nonincreasing along vertex cutsets away from an exogenous input. We use this spatial analysis to characterize signal-to-noise ratios (SNRs) in diffusive networks driven by process noise, and to develop a notion of propagation stability for dynamical networks. Finally, the formal results are illustrated through an example.

\end{abstract}


\section{Introduction}

The literature on dynamical networks has predominantly focused on emergence and control of global internal properties such as synchronization or consensus \cite{pecora,renbeard}. Recently, several studies considered input-to-output or transfer characteristics of dynamical networks. In particular, input-output analyses have been motivated by questions regarding the disturbance responses of networks as they are scaled \cite{studli2017vehicular,besselink2018scalable,mirabilio2021scalable}; these efforts generalize notions of string stability \cite{swaroop1996string} toward general network structures. In a separate track, input-to-output analyses have also been motivated by controller design needs in infrastructure networks with sparse actuation and measurement capabilities \cite{pirani,kooreh,mahia,abad2014graph}. Additionally, questions related to security and estimability of network processes have motivated input-output analyses \cite{vosughi,pasqualetti}.

Prior input-to-output analyses of dynamical networks have been concerned with particular metrics arising from application contexts, such as $l_\infty$ gains or transfer function zeros, and have primarily been focused on tying the input-output metrics to global graph properties (e.g. \cite{abad2014graph}). Relative to these efforts, our main contribution in this short article is to demonstrate that a compendium of input-output metrics exhibit a special spatial pattern, for a canonical network synchronization/consensus model. In particular, we study a standard discrete-time model for consensus/synchronization among nodes with scalar states \cite{xiao2004fast}. For this model, we consider several input-to-output metrics between node pairs, including $l_p$ gains, frequency responses, frequency-band energy, and Markov parameters. The main outcome of our analyses is to show that these metrics satisfy a {\em decrescence property}, wherein they are nonincreasing in value across cutsets away from an input node.  Additionally, two implications of the spatial pattern analysis are developed in brief vignettes. First, the analysis is used to give some insights into signal-to-noise ratios in network measurements, which bear on estimation/detection and sensor placement. Second, the analysis is used to define a notion of input-output or propagation stability for networks, and to verify propagation stability for the canonical model regardless of model scale (size).

\section{Model and Goals}

A dynamical network with $n$ agents or nodes is considered. Each node $i=1,\hdots, n$ has a scalar state $x_i(k)$ which evolves in discrete time ($k=0,1,2,\hdots$), via linear diffusive interactions with other nodes.  Here, our primary interest is in characterizing transfer characteristics in the network, hence we also model an exogenous input $u(k)$ applied at one source node $s$ and an output $y(k)$ taken at a target node $i$.  In particular, the following single-input single-output dynamics is considered:
\begin{eqnarray}
& & X(k+1) = AX(k) + e_{s}u(k) \label{eq:1}\\
& & y(k) = e_{i}^TX(k) \nonumber
\end{eqnarray}
where $X(k)=\begin{bmatrix} x_1[k] \\ \vdots \\ x_n[k] \end{bmatrix}$, $A = [a_{ij}]$ is  assumed to be row stochastic or substochastic (i.e. $a_{ij}\geq 0$ and $\sum_{j} a_{ij} \le 1$),
and the notation $e_b$ is used for a $0--1$ indicator vector whose $b$th entry is equal to $1$.  The zero-state response of the system (\ref{eq:1}) is of interest in this study.  The model (\ref{eq:1}) is a standard model for network synchronization or consensus (for stochastic $A$), and also is representative of other diffusive processes in networks.

A digraph $\Gamma$ is used to represent pairwise interactions among the nodes in the network.  Specifically, the graph $\Gamma$ is defined to have $n$ vertices labeled $1,\hdots, n$, which correspond with the $n$ network nodes. A directed edge is drawn from vertex $j$ to vertex $l$ if $a_{lj}>0$, which indicates that the state of vertex $j$ directly influences that of vertex $l$.  The vertices $s$ and $i$ are referred to as the source and target vertices, respectively. To simplify the development, we assume throughout that $\Gamma$ is strongly connected, i.e., there is a directed path from each vertex to every other vertex.  The strongly-connected case is of primary interest when considering transfer characteristics; although the results can be simply generalized for non-strongly connected cases, these details are omitted.

Our primary aim in this study is to compare input-output metrics for the network model (\ref{eq:1}) for different target vertex locations $i$ in $\Gamma$, so as to characterize the spatial pattern of the input-output response.  The specific metrics that we consider for the system (\ref{eq:1}) for a particular target location $i$ are:
\begin{itemize} 
\item The {\bf $l_p$ gain} $G_p(i,k_f)$ over a time horizon $k_f$, defined as
the maximum of the quantity $\left[\sum_{k=0}^{k_f} |y(k)|^p\right]^\frac{1}{p}$ over all inputs 
$u(0),\hdots, u(k_f-1)$ such that $\left[\sum_{k=0}^{k_f} |u(k)|^p\right]^\frac{1}{p}=1$.  
\item The {\bf frequency response} $H_{i}(e^{j\Omega})$, where $\Omega$ is the frequency of the (discrete time) input and $H_i()$ is the system's transfer function. Additionally, the responses to arbitrary periodic inputs, and the signal content in a frequency band, are also characterized.
\item The {\bf Markov parameters} $M_i(k)={\bf e}_i^T A^k {\bf e}_s$ for $k=0,1,2,\hdots$. 
\end{itemize}
 These metrics together form the basis for the external stability analysis of linear systems, and modulate estimation and feedback controller design.

Conceptually, the diffusive structure of the dynamics (\ref{eq:1}) suggests that inputs should have a localized sphere of influence, and hence the input-output metrics should exhibit a spatial degradation with distance from the source. Our goal is to provide a formal characterization of this spatial falloff, and to develop implications of this spatial analysis.

\section{Main Results: Spatial Pattern Analysis}

Graph theoretic analyses are developed for the input-output metrics defined for the system (\ref{eq:1}). The main results show that the metrics ($l_p$ gain, frequency response, Markov parameters) for different target locations falloff monotonically along graph cutsets away from the source location.
To formalize these notions, it is convenient to define the notion of a {\bf separating cutset} for a  graph.  To do so, let us consider a network graph $\Gamma$ with a source vertex $s$ and a particular target vertex $i=q^*$.  A set of vertices $C=\{ c(1),c(2),...,c(m) \}$ (which does not contain $s$ and $q^*$) is said to be a separating cutset, if any directed path from $s$ to $q^*$ passes through at least one vertex in $C$. The concept is illustrated in Figure \ref{fig:1}.  We also find it convenient to use the notation $Z$ for the partition of $\Gamma$ containing the target vertex, upon removal of the separating cutset. We refer to $Z$ as the {\bf target partition}.

%
\begin{figure}[!htb]
    	\centering
    	\vspace{-0.2cm}
    	\includegraphics[width=0.45\textwidth]{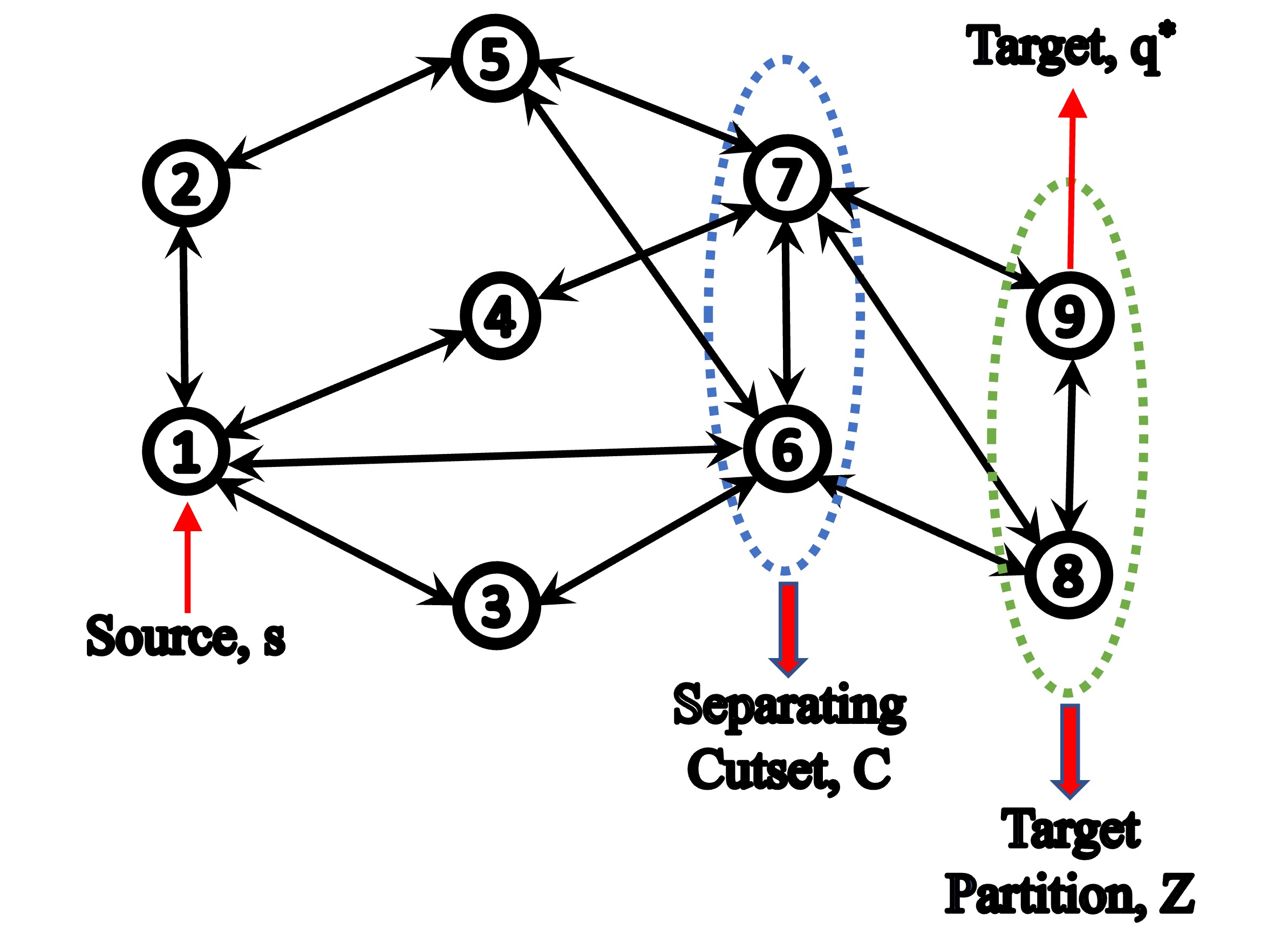}
    	\vspace{-0.2cm}
    	\caption{Illustration of a vertex cutset.}
    	\label{fig:1}
    	\vspace{-10 pt} 
\end{figure}
%


The spatial pattern analyses of the input-output metrics depend on a reformulation of system (\ref{eq:1}), wherein the target node's state dynamics are expressed in terms of the states of nodes corresponding to a separating cutset, rather than directly in terms of the input.
To develop this reformulation, let us define the {\bf cutset state vector} $X_c(k)$ as containing the time-$k$ states of the network nodes corresponding to the vertices in $C$.  Likewise, we define the {\bf target partition state vector} $X_z(k)$ as containing the time-$k$ states corresponding to the vertices in $Z$ (including $q^*$). The dynamical update of the target partition state vector can be expressed in terms of only the current cutset state and target partition state:
\begin{equation}\label{eq:2}
X_{z}(k+1) = A_zX_z(k) + B X_c(k)
\end{equation}
where $A_z$ is the principal submatrix of $A$ formed by selecting the rows and columns of $A$ indicated in $Z$; and $B$ is a submatrix of $A$ formed by selecting the rows indicated in $Z$ and the columns indicated in $C$.


The main graph-theoretic results build on the following characterization of the reformulated state dynamics (\ref{eq:2}):

\begin{lemma}
Assume that the network model (1) is initially relaxed and an exogenous input $u(k)$ is applied at the source node $s$.  Consider any separating cutset $C$ and corresponding target partition $Z$.
  The target partition state vector over the interval $k=0,1,\hdots, k_f$ can be expressed in terms of the cutset state vector as follows:

\begin{equation}\label{eq:3}
{\small \begin{bsmallmatrix} X_{z}(0)\\ \vdots \\ X_{z}(k_f) \end{bsmallmatrix}} = Q_{k_f}{\small \begin{bsmallmatrix} X_{c}(k_f)\\ \vdots \\X_{z}(0) \end{bsmallmatrix}},
\end{equation}

\noindent where 
\begin{equation}\label{eq:4}
Q_k = {\small \begin{bmatrix}
0&\dots&0&0\\
0&\dots&0&B\\
0&\dots&B&A_zB\\
\vdots&\ddots&\vdots & \vdots\\
0&\dots&A_z^{k-2}B&A_z^{k-1}B
\end{bmatrix}}.
\end{equation}
Further, the matrix $Q_{k_f}$ has nonnegative entries, and the sum of the entries in each row is at most $1$.
\end{lemma}



%
%


%
%


%
%


%
\begin{proof}

The target partition state vector can be computed in terms of the cutset state vector by solving (\ref{eq:2}).  This gives:
\begin{equation}\label{eq:5}
X_{z}(k) = \sum_{j = 0}^{k-1} A_z^{j}BX_c(k-j-1)
\end{equation}
Assembling the responses for $k=0,\hdots, k_f$ then immediately yields Equations (\ref{eq:3}) and (\ref{eq:4}) in the lemma statement.

%

%


%
%
%
%

The entries in $Q_{k_f}$ are nonnegative, as $A_z$ and $B$ are nonnegative.
To characterize the row sums for $Q_{k_f}$, let us first consider the last block of rows $\hat{Q}_{k_f} = [0  \quad B \hdots A_z^{k_f-2}B  \quad  A_z^{k_f-1}B]$. We characterize the row sums of $\hat{Q}_{k_f}$ by considering powers of the following matrix $F$:

\begin{equation*}
F = {\small \begin{bmatrix} A_z&B\\
        \mathbf{0} & I_m\end{bmatrix}}
\end{equation*}
The $k_f$th power of matrix $F$ is given by:
\begin{equation*}
F^{k_f} = {\small \begin{bmatrix} A^{k_f}_{z}& \sum_{k=0}^{k_f-1}A^{k}_{z}B\\
        \mathbf{0} & I_m\end{bmatrix}}
\end{equation*}

Since $[A_z \quad B]$ is a submatrix of $A$, it is immediate that $F$ is a stochastic or substochastic matrix. Consequently, the row sums of $F^{k_f}$ are also at most $1$. From the expression for $F^{k_f}$, it follows that the matrix $\sum_{k=0}^{k=k_f-1} A_z^kB$ has row sums of at most $1$. However, the row sums of $\sum_{k=0}^{k=k_f-1} A_z^kB$ and the row sums of the matrix $\hat{Q}_{k_f}$ are identical. Hence, the row sums of $\hat{Q}_{k_f}$ are upper bounded by $1$. Further, the sum of each row of the matrix $Q_{k_f}$ is upper bounded by one of the row sums of $\hat{Q}_{k_f}$.  The result thus follows.


\end{proof}

Our first main result is a graph-theoretic characterization of the $l_p$  gains of the input-output model (\ref{eq:1}):

\begin{theorem}
Consider the $l_p$ gain of the network model  (\ref{eq:1}) for a particular target node $i=q^*$, i.e. $G_p(q^*,k_f)$. Also consider the $l_p$ gains when the target node is alternately on a  separating cutset $C=\{ c(1),c(2),...,c(m) \}$, i.e. $G_p(c(i),k_f)$. For any time horizon $k_f$, it holds that $G_p(q^*,k_f)\leq G_p(c(i),k_f)$ for some $i=1,2,\hdots,m$.
\end{theorem}

\begin{proof}

%
%
%
%


Theorem 1 is verified for finite $p$ (i.e., $1 \le p < \infty$), and then for $p \to \infty$.


\begin{case}
$1\leq p<\infty$
\end{case}

The state of any node $q$ other than the source $s$ evolves as:
\begin{equation}\label{eq:8}
x_{q}(k+1) = a_{qq}x_q(k) + \sum_{j \in N(q)} a_{qj}x_{j}(k)
\end{equation}
\noindent where $N(q)$ contains the upstream neighbors of vertex $q$ in  $\Gamma$ (the vertices with directed edges to $q$). 
Notice that $a_{qj} > 0$ for $j \in {\cal N}(q)$ and $\sum_{j=1}^{n} a_{qj} \leq 1$.


%
%

%
%


%
%





For any $1\leq p< \infty$, the function $f(z) = |z|^p$ is convex over all $z$. Considering $z = x_q(k+1)$, it follows immediately that 

\begin{equation}
\begin{split}
|x_{q}(k+1)|^p \leq a_{qq}|x_q(k)|^p + \sum_{j \in N(q)} a_{qj}|x_{j}(k)|^p
\end{split}
\end{equation}

Summing the inequality for $k=0,\hdots, k_f$ yields:
\begin{equation}\label{eq:10}
\sum_{k = 0}^{k_f}|x_{q}(k+1)|^p \leq a_{qq} \sum_{k = 0}^{k_f} |x_q(k)|^p +  \sum_{j \in N(q)} a_{qj}\sum_{k = 0}^{k_f} |x_{j}(k)|^p
\end{equation}
Rewriting the leftmost sum and using $x_q(0)=0$, we get:
\begin{multline}\label{eq:11}
|x_{q}(k_f+1)|^p+\sum_{k = 0}^{k_f}|x_{q}(k)|^p \leq a_{qq} \sum_{k = 0}^{k_f} |x_q(k)|^p +  \\
\sum_{j \in N(q)} a_{qj}\sum_{k = 0}^{k_f} |x_{j}(k)|^p
\end{multline}

As $|x_{q}(k_f+1)|^p$ is nonnegative, we then obtain:
\begin{equation}\label{eq:12}
\sum_{k = 0}^{k_f}|x_{q}(k)|^p \leq \sum_{j \in N(q)} \frac{a_{qj}}{1-a_{qq}}\sum_{k = 0}^{k_f} |x_{j}(k)|^p
\end{equation}

Let us define $W_j = \frac{a_{qj}}{1-a_{qq}}$. Notice that $W_j > 0$ and $\sum_{j = 1}^{n} W_j \leq 1$. In this notation, the equation (\ref{eq:12}) can be written as:

\begin{equation}\label{eq:13}
\sum_{k = 0}^{k_f}|x_{q}(k)|^p \leq \sum_{j \in N(q)} W_{j}\sum_{k = 0}^{k_f} |x_{j}(k)|^p
\end{equation}

The term $\sum_{k = 0}^{k_f}|x_{j}(k)|^p$ is the $p$th power of the $p$-norm of the signal $x_{j}(k)$ over the interval $[0,k_f]$. Since equation (\ref{eq:13}) holds for any input, it follows that the $l_p$ gains satisfy the following for any node $q$ other than the source $s$:
\begin{equation}\label{eq:14}
(G_p(q,k_f))^p \leq \sum_{j \in N(q)} W_{j} (G_p(j,k_f))^p
\end{equation}

Now consider the target vertex $q^*$. From equation (\ref{eq:14}), $G_p(q^*,k_f) \leq G_p(j,k_f)$ for some $j\in N(q^*)$, and further the inequality is strict unless $G_p(j,k_f)=G_p(q^*,k_f)$ for all $j \in N(q^*)$. Then choose a neighbor $\hat{j}$ for which $G_p(j,k_f)$ is maximum. Repeating this argument for $\hat{j}$, we see that there is a vertex other than $q^*$ and $\hat{j}$, say $\overline{j}$, such that $G_p(\overline{j},k_f) \leq G_p(\hat{j},k_f)$. Iterating,  we finally get that $G_p(q^*,k_f)\leq G_p(c(i),k_f)$ for some $i=1,2,\hdots,m$.

\begin{case}
$p \rightarrow \infty$
\end{case}

The proof follows readily from Lemma 1. Specifically, consider the separating cutset $C$ and the resulting target partition $Z$.  Applying Equation (\ref{eq:3}) in Lemma 1 and then taking the infinity norm of both sides of Equation (\ref{eq:3}), we obtain:
%
%
%
%
%
\begin{equation}\label{eq:15}
\|{\small \begin{bmatrix} X_z(0) \\ \vdots \\ X_z(k_f) \end{bmatrix}}\|_\infty \leq \|Q_{k_f}\|_\infty
\|{\small \begin{bmatrix} X_c(k_f)\\ \vdots \\ X_c(0) \end{bmatrix}}\|_\infty
\end{equation}
The maximum row sum of $Q_{k_f}$ is $1$ and hence $ \|Q_{k_f}\|_\infty \leq 1$. As a result, we get the following inequality:
\begin{equation}\label{eq:16}
\|{\small \begin{bmatrix} X_z(0) \\ \vdots \\ X_z(k_f) \end{bmatrix}}\|_\infty \leq
\|{\small \begin{bmatrix} X_c(k_f)\\ \vdots \\ X_c(0) \end{bmatrix}}\|_\infty
 \end{equation}
Equation (\ref{eq:16}) holds for all inputs $u(k)$, including the input that maximizes the $\infty$-norm when the target node is $q^*$.  Since $q^* \in Z$, it follows immediately from (\ref{eq:16}) that $G_p(q^*,k_f)$ is upper bounded by the maximum magnitude entry in $ \begin{bmatrix} X_c(k_f)\\ \vdots \\ X_c(0) \end{bmatrix}$ for this input.  Therefore,  $G_p(q^*,k_f)\leq G_p(c(i),k_f)$ for some $i=1,2,\hdots,m$.


\end{proof}

\begin{remark}
The spatial pattern on the $l_p$ gains also holds in the infinite-horizon limit ($G_p(q^*,k_f)$ for $k_f \rightarrow \infty$), provided that the gain is finite.  A finite gain is achieved if $A$ is strictly substochastic (i.e., at least one row sum is strictly less than $1$), since the matrix $A$ is irreducible (the graph $\Gamma$ is strongly connected) by assumption.
\end{remark}

\begin{remark}
Per Theorem 1, the spatial degradation of the $l_p$ norms holds for model (\ref{eq:1}) for any input signal. It follows that the spatial result of the $l_p$ gains also holds for the closed-loop system, when a feedback is applied from the target to the source node.  
\end{remark}

\begin{remark}
The proof also trivially extends to mixed-norm gains, since the spatial pattern holds for any input. 
\end{remark}

Next, the response of the input-output model (Equation 1) to a periodic input, and hence the frequency response, is also verified to exhibit a spatial degradation:

\begin{theorem}
Consider the zero-state (i.e. $X(0)=0$) response of the network input-output model (\ref{eq:1}), under the  assumption that the matrix $A$ is strictly substochastic (i.e. at least one row sum is strictly less than $1$).  Consider a particular target vertex $i=q^*$. Also, consider any separating  cutset $C=\{ c(1),c(2),...,c(m) \} $. Then the following statements are true:

\begin{enumerate}[label=(\roman*)]
\item For a periodic input $u(k)$, the response at the target vertex $q^*$, $x_{q^*}(k)\leq \max \{0, \max_{\widehat{k}=0,1,\hdots,k} x_{c(i)}(\widehat{k})\}$ for some $i=1,\hdots,m$. Also, $x_{q^*}(k)\geq \min \{0, \min_{\widehat{k}=0,1,\hdots,k} x_{c(j)}(\widehat{k})\}$ for some $j=1,\hdots,m$.

\item At each frequency $\Omega$, $|H_{q^*}(e^{j\Omega})| \leq |H_{c(i)}(e^{j\Omega})|$ for some $i=1,\hdots, m$, where $H_{b}(j \omega)$ is the frequency response when the output is taken at vertex $b$.
\end{enumerate}

\end{theorem}

\begin{proof}

Since $A$ is substochastic and also irreducible ($\Gamma$ is strongly connected) by assumption, the system is asymptotically stable in the sense of Lyapunov. Hence, the response to a periodic input is periodic in the asymptote, and bounded for all time. 


Let us consider the solution of equation (\ref{eq:2}) at time $k$:
\begin{equation}\label{eq:17}
X_{z}(k) = {\small \begin{bmatrix} 0&B&A_zB& \hdots A_z^{k-1}B \end{bmatrix}} {\small \begin{bmatrix} X_c(k)\\\vdots\\ X_c(0) \end{bmatrix}}
\end{equation}

Since the target vertex $q^*$ is in the target vertex set $X_z(k)$, we can write the corresponding response as follows:
\begin{equation}\label{eq:18}
x_{q^*}(k) = {\small \begin{bmatrix} 0&B&A_zB& \hdots A_z^{k-1}B \end{bmatrix}}_{q^*} {\small \begin{bmatrix} X_c(k)\\\vdots\\ X_c(0) \end{bmatrix}}
\end{equation}
where ${\small \begin{bmatrix} 0&B&A_zB& \hdots A_z^{k-1}B \end{bmatrix}}_{q^*}$ is the row of ${\small \begin{bmatrix} 0&B&A_zB& \hdots A_z^{k-1}B \end{bmatrix}}$ corresponding to the target vertex.

Notice that the row sum of ${\small \begin{bmatrix} 0&B&A_zB& \hdots A_z^{k-1}B \end{bmatrix}}$ is less than or equal to unity. Therefore, it is immediate that for a periodic input, $x_{q^*}(k)\leq \max \{0, \max_{\widehat{k}=0,1,\hdots,k} x_{c(i)}(\widehat{k})\}$ for all $k$ and some $i=1,\hdots,m$. Similarly, the lower bound of the target state can be characterized in terms of cutset state as $x_{q^*}(k)\geq \min \{0, \min_{\widehat{k}=0,1,\hdots,k} x_{c(i)}(\widehat{k})\}$, for some $i=1,\hdots,m$.



To prove part (ii), consider the response of system (\ref{eq:1}) for a sinusoidal input $u(k) = \cos{\Omega k}$. The response at the each node has a sinusoidal steady-state component at the driving frequency $\Omega$, plus a transient component.  The sinusoidal component of the response at the target node is given by
\begin{equation}\label{eq:19}
x_{q^*,SS}(k) = |H_{q^*}(e^{j\Omega})|\cos({\Omega k + \angle{H_{q^*}(e^{j\Omega})}}).
\end{equation}
Similarly, the sinusoidal component of the response at each separating cutset node is given by:
\begin{equation}\label{eq:19a}
x_{c(i),SS}(k) = |H_{c(i)}(e^{j\Omega})|\cos({\Omega k + \angle{H_{c(i)}(e^{j\Omega})}}).
\end{equation}

Now consider finding the response at the target node using Equation (\ref{eq:2}), when the driving signal $X_c(k)$ is set to contain  the sinusoidal components of the cutset nodes' state $x_{c(i),SS}(k)$ for the sinusoidal input.  The response at target node computed in this way has a transient and a sinusoidal steady-state component.  Importantly, the sinusoidal component is identical to the sinusoidal response of the original network model (\ref{eq:1}) for the input $u(k)=\cos{\Omega k}$, i.e. the response is $x_{q^*,SS}(k)=|H_{q^*}(e^{j\Omega})|\cos({\Omega k + \angle{H_{q^*}(e^{j\Omega})}})$.

Finally, using the same argument as for other periodic inputs, one finds that  $x_{q^*,SS}(k) \le \max_{\widehat{k}=0,1,\hdots,k} x_{c(i),SS}(\widehat{k})$
for some $i=1,\hdots, m$, for all sufficiently large $k$.  However, this is only possible if $|H_{q^*}(e^{j\Omega})| \leq |H_{c(i)}(e^{j\Omega})|$ for some $i=1,\hdots, m$.

%



%



\end{proof}

\begin{remark}

The result in Theorem 2(i) is presented in a generalized form, to account for any periodic input. If the maximum of the responses over the cutset vertex set is positive, then the comparison with $0$ can be excluded in the statement for the maximum response. Similarly, if the minimum of the responses over the cutset vertices is negative, the comparison with 0 can be excluded in the statement for the minimum.
\end{remark}

\begin{remark}

The zero ($0$) terms in the comparisons in Theorem 2(i) can be excluded, if the state matrix $A$ is stochastic and further the asymptotic response is considered. That is, the upper and lower bound of the target is dependent only on the response of the cutset vertices in this case; this can be verified 
by using the fact that the row sums of ${\small \begin{bmatrix} 0&B&A_zB& \hdots A_z^{k-1}B \end{bmatrix}}$ approach unity asymptotically.

\end{remark}

The frequency-response analysis carries through to the case that $A$ is stochastic, i.e. its row sums are unity, with only a couple of slight limitations. This formalized in the following theorem:

\begin{theorem}
Consider the network model (\ref{eq:1}) in the case that $A$ is an ergodic stochastic matrix. Also, consider a target node $q^*$,  and any separating cutset $C=\{ c(1),c(2),...,c(m) \} $. Then, for each frequency $\Omega>0$, $|H_{q^*}(e^{j\Omega})| \leq |H_{c(i)}(e^{j\Omega})|$ for some $\Omega>0$. .
\end{theorem}

\begin{proof}

Since $A$ is assumed ergodic and stochastic, it has a single eigenvalue at unity with remaining eigenvalues strictly inside the unit circle.  It follows that the response of the system (\ref{eq:1}) to a sinusoidal input at a non-zero frequency $\Omega$ produces a bounded sinusoidal response at the same frequency, i.e. the frequency response is finite and well-defined at $\Omega$. With this observation, the remainder of the proof is identical to that of Theorem 2.



\end{proof}

Theorems 2 and 3 have characterized the spatial response properties of the system (\ref{eq:1}) at a particular frequency. The energy in the response over a frequency band can also be shown to exhibit a spatial falloff, by relying on Parseval's theorem and the spatial analysis of the two-norm response.  This result is formalized in the following theorem:


%
\begin{theorem}
Consider the system (\ref{eq:1}), under the assumption that $A$ is substochastic. Also, consider a separating cutset $C=c(1),c(2),...,c(m)$ (where the target node is $q^*$). Then, $\int_{\Omega_1}^{\Omega_2} |H_{q^*}(e^{j\Omega})|^2 d\Omega \leq \int_{\Omega_1}^{\Omega_2} |H_{c(i)}(e^{j\Omega})|^2 d\Omega$ for some $i=1,2,\hdots,m$, and any frequency band $[\Omega_1 \quad \Omega_2]$.
\end{theorem}

\begin{proof}

The result is proved based on the norm inequality (\ref{eq:13}), for the special case where $p=2$. Equation (\ref{eq:13}) in this case is give by:

\begin{equation}\label{eq:20}
\sum_{k = 0}^{k_f}|x_{q}(k)|^2 \leq \sum_{j \in N(q)} W_{j}\sum_{k = 0}^{k_f} |x_{j}(k)|^2
\end{equation}

Now, consider the target vertex $q^*$ and the vertex cutset $C = c(1), c(2), \hdots, c(m)$. Using similar arguments to those used in the proof of theorem 1, we can obtain the following inequality:
\begin{equation}\label{eq:21}
\sum_{k = 0}^{k_f}|x_{q^*}(k)|^2 \leq \sum_{k = 0}^{k_f} |x_{c(i)}(k)|^2,
\end{equation}
for some $i=1,\hdots, m$.  The inequality holds for any input $u(k)$, and any $k_f$. Considering $k_f \to \infty$ and then applying Parseval's theorem (which holds since the system is asymptotically stable), Equation (\ref{eq:21}) can be written in the frequency domain as follows:

\begin{equation}\label{eq:22}
\frac{1}{2\pi} \int_{-\pi}^{\pi}|X_q(e^{j\Omega})|^2d\Omega \leq \frac{1}{2\pi} \int_{-\pi}^{\pi}|X_{c(i)}(e^{j\Omega})|^2d\Omega
\end{equation}

where we use the notation $X_i(e^{j\Omega})$ as the Fourier transform of the signal $x_i(k)$. 
Using the relationship $|X_i(e^{j\omega})| = |H_i(e^{j\Omega})||U(e^{j\Omega})|$. Equation (\ref{eq:22}) can be rewritten as:

%
\begin{equation}\label{eq:23}
\int_{-\pi}^{\pi}|H_q(e^{j\Omega})|^2|U(e^{j\Omega})|^2 d\Omega \leq \int_{-\pi}^{\pi}|H_{c(i)}(e^{j\Omega})|^2|U(e^{j\Omega})|^2 d\Omega
\end{equation}

As Equation (\ref{eq:23}) holds for any input, let us choose the input to be  $u(k) = \frac{1}{\pi k}\sin(Wk)e^{jW_0k}$. The Fourier transform of this input, $U_1(e^{j\Omega})$, is unity for $|\Omega-W_0| <W$ and zero otherwise.  Choosing $W_0-W=\Omega_1$ and $W_0+W=\Omega_2$ and substituting into (\ref{eq:23}), we find:
\begin{equation}\label{eq:25}
\int_{\Omega_1}^{\Omega_2} |H_{q^*}(e^{j\Omega})|^2 d\Omega \leq \int_{\Omega_1}^{\Omega_2} |H_{c(i)}(e^{j\Omega})|^2 d\Omega
\end{equation}
where $c(i)$ is a vertex in vertex cutset $C$.

\end{proof}

\begin{remark}
Since Equation (\ref{eq:23}) is valid for any input, the spatial degradation pattern also holds for any frequency-domain signal which is filtered by the network dynamics. Specifically, $\int_{\Omega_1}^{\Omega_2}|H_q(e^{j\Omega})|^2|F(e^{j\Omega})|^2d\Omega \leq \int_{\Omega_1}^{\Omega_2}|H_{c(i)}(e^{j\Omega})|^2|F(e^{j\Omega})|^2d\Omega$ for the target node $q^*$ and some $c(i)$ in vertex cutset $C$, where $F(e^{j\Omega})$ is an arbitrary function which is filtered by the network.
\end{remark}

Now, we show that the Markov parameters also exhibit a spatial degradation:

\begin{theorem}
Consider the Markov parameters for  the system (1), in the cases that the output is taken at the target vertex $i=q^*$ and at  vertices on a separating cutset $C=\{ c(1),c(2),...,c(m) \}$. The Markov parameters satisfy the following inequality: for all $k$,  $M_{q^*}(k)\leq \max_{\widehat{k}=0,1,\hdots, k}M_{c(i)}(\widehat{k})$ for some $i=1,2,\hdots,m$.
\end{theorem}

\begin{proof}

The proof of theorem 3 can be derived from the proof of theorem 2. Notice that equation (\ref{eq:18}) holds for any input. So, for any input it follows that $x_{q^*}(k)\leq \max_{k}x_{c(i)}(k)$ for all $k$ and some $i=1,2,\hdots,m$ in vertex cutset $C$. If the system is initially relaxed and driven by an impulse input, the response of target node $q^*$ can be written as $x_{q^*}(k) = e_{q^*}A^{k-1}e_s = M_{q^*}(k-1)$. Similarly, the response in a vertex cutset node $c(i)$ can be expressed as $x_{q^*}(k) = e_{c(i)}A^{k-1}e_s = M_{c(i)}(k-1)$. Form this, it is immediate that $M_{q^*}(k)\leq \max_{k}M_{c(i)}(k)$ for some $i=1,2,\hdots,m$.

\end{proof}

We have focused in this short article on the single-input single-output system (\ref{eq:1}), with the goal of understanding pairwise transfer characteristics in network synchronization or diffusion processes.  However, it turns out that the spatial pattern of responses in the synchronization model holds even when inputs are applied at multiple network nodes. To formalize this notion, we consider a modified model where inputs are applied at multiple nodes:
\begin{eqnarray}
& & X(k+1) = AX(k) + BU(k) \label{eq:26}\\
& &y(k) = e_{i}^TX(k) \nonumber
\end{eqnarray}
where $X(k)$, $A$, and $e_i$ are defined as before in Equation (\ref{eq:1}), $U(k)=\begin{bmatrix} u_1(k) \\ \vdots \\ u_{\hat{m}}(k) \end{bmatrix}$ is a vector of inputs, and  $B$ is a $n \times \hat{m}$ matrix whose columns are $0--1$ indicator vectors for a set of input nodes in the network (input vertices in the network's graph). 

As an illustration, we characterize the $l_p$ gain for the multi-input case.  This requires a redefinition of the gain concept:

\begin{definition}
The {\bf $l_p$ gain} $G_p(i,k_f)$ for the system (\ref{eq:26}) over the time horizon $k_f$, with output taken at node $i$, is defined as the  maximum of the quantity $\left[\sum_{k=0}^{k_f} |y(k)|^p\right]^\frac{1}{p}$ over all inputs $U(0),\hdots, U(k_f-1)$ such that $\left[\sum_{k=0}^{k_f} \sum_{j=0}^{\hat{m}} |u_j(k)|^p\right]^\frac{1}{p}=1$.
\end{definition}

The $l_p$ gains of the system in equation (\ref{eq:26}) also show a spatial falloff defined by cutsets in the network graph, as formalized in the following theorem:

\begin{theorem}
Consider the $l_p$ gain $G_p(q^*,k_f)$ of the system (\ref{eq:26}) for a target node $q^*$ and a time horizon $k_f$. Also, assume that $C=\{ c(1),c(2),...,c(m) \}$ is a separating cutset of the network, in the sense that all directed paths in $\Gamma$ from each input vertex to the target vertex pass through at least one of these vertices. Then, the $l_p$ gain of a cutset node majorizes the $l_p$ gain of the target, i.e. $G_p(q^*,k_f)\leq G_p(c(i),k_f)$ for some $i=1,2,\hdots,m$.
\end{theorem}
\begin{proof}

The proof of Theorem 4 is nearly  identical to that of Theorem 1, hence we omit the details.

%
%
%
%



\end{proof}

\section{Simulations}

The spatial fall-off in the $l_p$ gains (Theorem 1) is illustrated in an example.  A network with $9$ nodes with the following substochastic state matrix is considered:
\begin{equation*}
{\small A=\begin{bmatrix}
0&0&0&0&0&0&0.2&0.3&0.25\\
0&0&0&0&0&0&0.35&0.25&0.2\\
0&0&0&0&0&0&0.15&0.25&0.45\\
0.2&0.4&0.35&0&0&0&0&0&0\\
0.2&0.15&0.25&0&0&0&0&0&0\\
0.2&0.45&0.15&0&0&0&0&0&0\\
0&0&0&0.25&0.25&0.15&0&0&0\\
0&0&0&0.3&0.35&0.25&0&0&0\\
0&0&0&0.25&0.35&0.1&0&0&0\\
\end{bmatrix}}
\end{equation*}
The network's digraph, which has 27 directed edges, is shown in Figure 2.  The source vertex (Vertex 1) is also indicated.

	\begin{figure}[hbt!]
    	\centering
    	\includegraphics[width=0.48\textwidth]{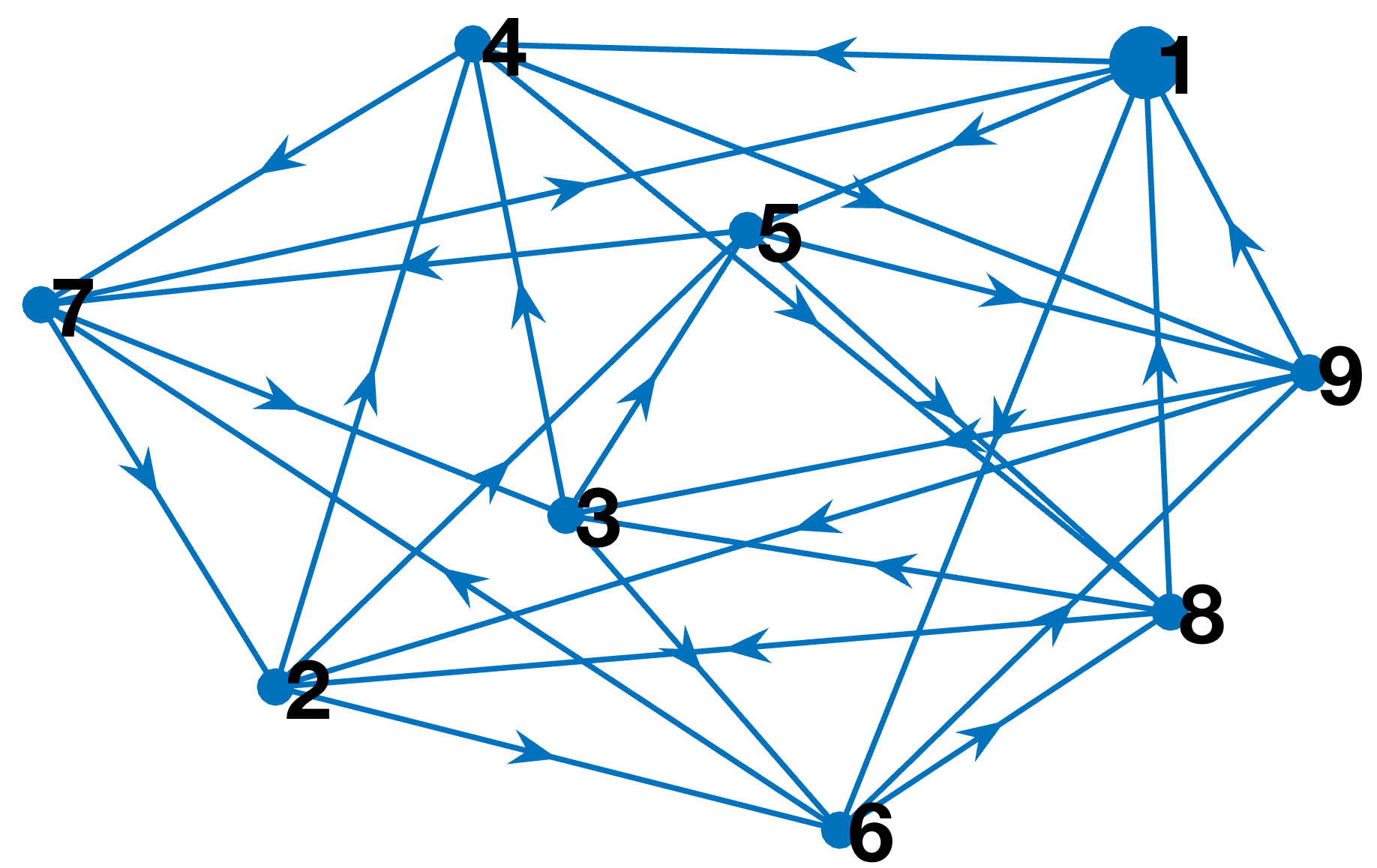}
    	\caption{The digraph of the example network.}
    	\label{fig:2}
	\end{figure}

To illustrate Theorem 1, the $l_1$ gain is computed when each node in the network is selected as the target.  These gains are plotted on top of the network's graph in Figure 3.  The plot shows the decrescence in the gains along cutsets in the network graph.  As one example, the $l_1$ gain when vertex $9$ is the target is upper-bounded by the maximum among the $l_1$ gains when vertices $4$, $5$, and $6$ are chosen as targets. Since these vertices form a separating cutset, the example matches with the formal result.

	\begin{figure}[!htb]
    	\centering
    	\includegraphics[width=0.48\textwidth, height=0.27\textwidth ]{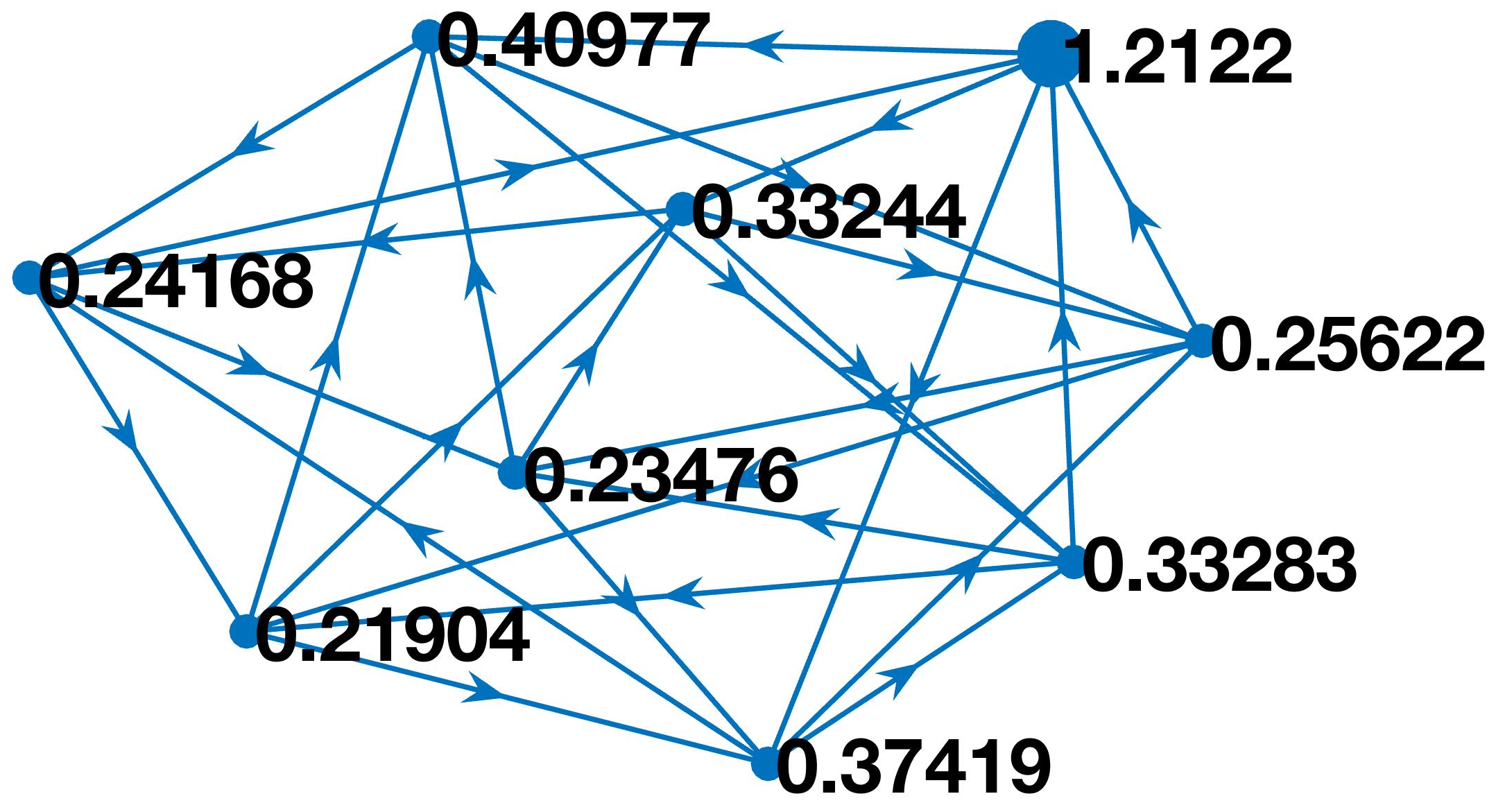}
    	\caption{The $l_1$ gains for different target nodes.}
    	\label{fig:4}
	\end{figure}

The decrescence of $l_p$ gains along cutsets also immediately implies that the maximum gain among target vertices at a certain distance from the source is a non-increasing function of the distance.  This is true because the set of vertices at each distance forms a separating cutset between the source vertex and more distant vertices.  Thus, it is insightful to calculate maximum $l_p$ gains for target vertices at different distances.  We have done so for the $1$, $2$, and $\infty$ gains in Table I. Each $l_p$ gain is highest at the source (distance $0$), and decreases with distance from source.

\begin{table}[!htb]
    \centering
\begin{tabular}{|c|c|c|c|}
    \hline
     Distance & Max $l_1$ Gain & Max $l_2$ Gain & Max $l_\infty$ Gain\\
     \hline
     0  &  1.2122 &    1.1676 &   1.2122\\
     \hline
     1 &    0.4098 &   0.3654  &  0.4098\\
     \hline
     2 &    0.3328  &  0.2994  &  0.3328\\
     \hline
     3 &    0.2348 &   0.2113  &  0.2348\\
     \hline
\end{tabular}
\caption{Decrescence of $l_p$ gains with distance from the source.}
\end{table}

\section{Applications of the Spatial Pattern Analysis}

Two brief vignettes are presented, to illustrate potential applications of the spatial analyses developed in the short paper.  The first is concerned with signal recovery from remote measurements in a network, while the second is concerned with definitions of (external) propagation stability in networks.  

\subsection{Signal Recovery and Signal-to-Noise Ratios in Diffusive Networks}

The problem of estimating or detecting an external input to a network synchronization or diffusion process from noisy remote measurements arises in several settings\cite{ye2018optimal}.  For instance, it may be necessary to locate and characterize a pollution source in a complex environment, based on localized measurements of the diffusing pollutant.  Similarly, it may be of interest to identify a stubborn or manipulative agent in a network opinion dynamics, based 
on imperfect measurements of certain agents' opinions.  Standard machinery for detection/estimation (e.g., hypothesis testing, Wiener filtering) can be brought to bear for these problems.  In turn, the statistical performance of the recovery technique can be characterized, and used to develop algorithms for sensor placement.  However, these analyses -- particularly ones for performance analysis and sensor placement  -- can become challenging to develop and/or computationally infeasible, particularly when the network is high dimensional or incompletely known.  For these reasons, graph-theoretic insights into network signal estimation/detection performance are valuable.

Broadly, the performance of estimators and detectors is often closely tied with the signal-to-noise ratio (SNR) in the measured signal, i.e. the ratio between the energy/power of the signal of interest and that of the noise.  The characterizations of input-output responses developed in this article imply that signal-to-noise ratios also degrade spatially in a network away from a signal source, and hence suggests that detection/estimation performance degrades monotonically away from the source.  

To present the concept more explicitly, let us consider the network dynamics (\ref{eq:1}) when the input $u(k)$ at the source node is an unknown signal, and measurements at a remote target node are subject to an additive zero-mean white noise with variance $\sigma^2$.  That is, we assume that measurements taken at a node $i$ are of the form $\overline{y}(k)=x_i(k)+v(k)$, where $v(k)$ is a zero-mean white noise signal.  

In the case where the the input $u(k)$ is a finite-energy or transient signal, it is natural to define the SNR at the target node as the ratio between the energy (or other norm) of the response signal $x_i(k)$ and the noise intensity, e.g. $SNR=\frac{||x_i(k)||_p}{\sigma}$.  From Theorem 1, it is immediate that the SNR for any input signal (and the best-case SNR among all inputs with a certain energy/norm) exhibits the spatial degradation pattern defined by graph cutsets.  That is, as compared to the SNR at a target node, the SNR will be higher for at least one node on a cutset separating the source and target.  Thus, estimation or detection performance should be improved on the separating cutset, i.e. at nodes closer to the source.

In the case where the input $u(k)$ is a persistent signal, the SNR is typically defined in terms of the power of the response signal $x_i(k)$
and the noise.  A typical case may be that the input and hence the the response signal $x_i(k)$ is a stochastic band-limited signal.  In this case, the total power in $x_i(k)$ can be determined as the integral of the power spectrum over the band.  Thus, the SNR can be found as:
$SNR=\frac{\int_{\Omega}|H_i(e^{j\Omega})|^2 S_{UU}(\Omega)| d\Omega}{\sigma^2}$, where $S_{UU}(\Omega)$ is the power spectrum of the input signal $u(k)$.  From Theorem 3 and the following remark (Remark 2), we immediately see that the SNR in this case also exhibits a spatial degradation along cutsets away from the source.  Thus, estimation or detection is again seen to be improved at cutsets near the source node.

We stress that the spatial degradation of signal strengths and SNRs holds not only for the $2$-norm but for all other $p$-norms, which are relevant for a number of signal reconstruction and detection problems \cite{pnormsignal}.

\subsection{Network Propagation Stability Analysis}

The stability analysis of network synchronization models has been primarily focused on internal (Lyapunov) stability notions, which capture emergence of a synchronous state \cite{renbeard,pecora,li2006global}. However, in many application areas, the responses of network processes to exogenous disturbances is of substantial interest \cite{lu2011stabilization,liu2013input,siami2016fundamental}.  With this motivation,  a number of studies have considered external stability of network synchronization models \cite{gchen,wangoutput1}. These studies generally define external stability as a bounded-input bounded-output (BIBO) stability notion, with the assumption that exogenous (disturbance) inputs are present at any node or a set of leader nodes, and the output is the full network's state or projections thereof.

For many network processes, a primary concern is whether a disturbance at one node can amplify and propagate across the network, or whether alternately it remains localized. For instance, power-system operators recognize that oscillatory disruptions originating from generator controls can couple with the system's natural modes to cause system-wide oscillations. The propagation or amplification of a disturbance across a network is distinct from both the internal stability and external stability concepts. This distinction motivates alternative definitions of stability for networks, concerned with disturbance propagation.

In fact, propagation of disturbances in cascaded systems and bidirectional line networks has been extensively studied, under the heading of {\em string stability} \cite{peppard1974string,7879221,swaroop1996string}. Several subtly different definitions of string stability have been proposed. Broadly, the definitions require that a disturbance originating at one location in the string can only be amplified by a specified finite gain for any output location, regardless of the length of the string.  A stronger notion, referred to as {\em strict string stability},  requires that the disturbance response is diminished (scaled down) at each subsequent node in the string \cite{ploeg2013lp,rogge2008vehicle}. The string stability notion has been extended for other network topologies in a couple of ways.  First, there is a considerable literature on {\em mesh stability}, which generalizes the string-stability concept to general directional networks \cite{983389,9449895}. Also, an important recent study by Studli and co-workers \cite{studli2017vehicular} has considered generalization of string stability concepts to more general networks.  This study provides a parallel definition to string stability for general networks, as well as a parallel to strict string stability for tree-like networks, although without characterizations for particular network models.   Other input-to-state stability concepts with a similar flavor, which generalize mesh stability, have also recently been developed for networks with general graph topologies and nonlinear nodal dynamics \cite{8370724}.

We argue that our spatial input-output analysis immediately suggests a definition for strict stability with respect to disturbance propagation, for general network processes.  Specifically, a network dynamics can be viewed as strictly propagation stable, if responses degrade along cutsets away from the source.  This notion is formalized in the following definition:

\begin{definition}
Consider a dynamical network process where an input ${\bf u}(k)$ is applied at a source node $s$, and the state responses ${\bf y}_{q^*}(k)$ are considered at target nodes $q^*$ in the network.  The network is ${\cal L}_{p,t}$ strictly propagation stable if, for every source node $s$, target node $q^*$, separating cutset $C=\{ c(1),c(2),...,c(m) \}$, and finite-$p$-norm input ($||u(k)||_p < \infty$), it holds that $||{\bf y}_{q^*}(k)||_t \leq 
||{\bf y}_{c(i)}(k)||_t$ for some $i=1,\hdots,m$.
\end{definition}

The definition for propagation stability can easily be rephrased in terms of paths in the network's graph.  In particular, a network is ${\cal L}_{p.t}$ strictly propagation stable if and only if, for every source vertex $s$, target vertex $q^*$, and finite-$p$-norm input, there exists a path from $s$ to $q^*$ in the network graph such that the $t$-norm of the response at the corresponding network nodes is decreasing along the path.  From this phrasing, it is evident that the concept reduces to the standard notions for strict stability for strings and tree graphs.  The statement also clarifies that the definition is flexible to the complex propagations that may occur in networks, in that decrescence is only required along one path  between a source and target node (not all paths).

From the proof of Theorem 1, the scalar network synchronization model (\ref{eq:1}) is immediately seen to be ${\cal L}_{p,t}$ strictly propagation stable.  Thus, we see that disturbances in this model are localized/diminishing in the sense of the strict propagation stability definition, regardless of the specific network topology. For more complex synchronization models involving multivariate nodal dynamics, one anticipates that the node or subsystem model and the network graph will together determine whether or not strict propagation stability holds.

\section{Conclusions}

This short note has focused on the input-to-output analysis of a canonical discrete time network synchronization model.
Our key finding is the a family of input-to-output metrics (e.g. $l_p$ gains, frequency responses, frequency-band energy, Markov parameters) show a spatial degradation across cutsets away from the input location.  Two applications of the analysis have been briefly developed: 1) characterization of signal-to-noise ratios in network processes, and 2)  definition of propagation stability notions for dynamical networks. Of interest, the spatial analyses carry through to a number of time-varying and nonlinear processes.  We expect to address these cases in future work.


\bibliographystyle{IEEEtran}

\bibliography{main}

\end{document}